\documentclass[11pt]{article}
\usepackage{amsfonts}
\usepackage{setspace} 
\usepackage{amsthm}
\usepackage{amsmath}
\usepackage{amssymb}
\usepackage{latexsym}
\usepackage{bbm}
\usepackage[pdftex]{graphicx,color}
\usepackage[table,xcdraw]{xcolor}
\usepackage{enumerate}
\usepackage{verbatim}
\usepackage{amsfonts}
\usepackage{epsfig}
\usepackage{tikz}
\usepackage{hyperref}
\usepackage{caption} 
\usepackage{bbm, color} 
\usepackage{amssymb,amsmath,tikz,graphicx,float,xspace,chemarrow}
\usepackage{graphicx}
\usepackage[utf8]{inputenc}
\makeatletter

\newcommand{\be}{\begin{equation}}
\newcommand{\ee}{\end{equation}}
\newcommand{\bea}{\begin{eqnarray}}
\newcommand{\eea}{\end{eqnarray}}

\newcommand{\distas}[1]{\mathbin{\overset{#1}{\kern\z@\sim}}}%
\newsavebox{\mybox}\newsavebox{\mysim}
\newcommand{\distras}[1]{%
  \savebox{\mybox}{\hbox{\kern3pt$\scriptstyle#1$\kern3pt}}%
  \savebox{\mysim}{\hbox{$\sim$}}%
  \mathbin{\overset{#1}{\kern\z@\resizebox{\wd\mybox}{\ht\mysim}{$\sim$}}}%
}

\newcommand{\x}{\mathbf{x}}

\newtheorem{theorem}{Theorem}
\newtheorem{lemma}{Lemma}

\newtheorem{remark}{Remark}
\newtheorem{corollary}{Corollary}

\newtheorem{definition}{Definition}

\definecolor{LightCyan}{rgb}{0.88,1,1}
\definecolor{shadecolor}{rgb}{0.01,0.4,.8}
\begin{document}
\begin{center} { \bf \large \sc 
\bf \sc Predicting the scoring time in hockey}
\vskip 0.1in { Abdolnasser Sadeghkhani, Syed Ejaz Ahmed }
\\
\vskip 0.1in
Department of Mathematics, Brock University, St. Catharines, ON, Canada
\\ E-mail: asadeghkhani@brocku.ca, sahmed5@brocku.ca

\end{center}
\vspace*{0.5cm}

\normalsize
\begin{abstract}
In this paper, we propose a Bayesian predictive density estimator to predict the time until the $r$--th goal is scored in a hockey game, using ancillary information such as their performances in the past, points and specialists' opinions. To be more specific, we consider a gamma distribution as a waiting scoring model. The proposed density estimator belongs to an interesting new version of weighted beta prime distribution and outperforms the other estimator in the literature. The efficiency of our estimator is evaluated using frequentist risk along with measuring the prediction error from the old dataset, $2016-17$, to the current season ($2018-19$) of the National Hockey League.
\end{abstract}


\section{Introduction}
Predicting an unobserved random variable by finding the future density is often more meaningful for applications than conventional point and interval estimation of parameters, because it gives a richer prediction and all point or interval estimation, percentile, and even hypothesis testing can be obtained from the estimated density. Although in practice there is a plug--in approach in density estimation problems, Bayesian predictive distributions are a better alternative in the literature in many cases (see e.g. Marchand and Sadeghkhani 2018) and most of the time are easily calculated by using modern Monte Carlo techniques. 

Often, there exists some ancillary information at our disposal which has been unduly ignored. Perhaps the most straightforward way to visualize this additional information is to place restrictions on the parameters of our model and translate them into parametric restrictions. In this paper, we focus on the gamma distribution as a waiting scoring model. Previous studies indicate estimating of future density based on current observed gamma population. For example, L'moudden et. al. (2017) studied Bayesian estimation of a future random variable of gamma distribution when the scale parameter is bound within an interval.
This paper, uses ancillary information obtained from two independent gamma densities, to make a more accurate density estimation of one of them.

The remainder of the paper is organized as follows: In Section 2, we provide definitions and preliminary remarks which will be needed throughout this article. Section 3 discusses how to find a Bayesian predictive density estimator in gamma distribution with and without contemplating the ancillary information. In Section 4 we study an interesting application of proposed methods in estimating the future density of waiting time until scoring $r$-th goal in a hockey game. Finally, we make some concluding remarks in Section 5.

\section{Main results}\label{mainresult}
Let $X \mid \lambda \sim \operatorname{Gam}(r, \lambda)$, be a random variables (rv) from a gamma distribution with probability distribution function (pdf) 
\begin{align*}
p_{\lambda}(x)&=\frac{x^{r-1} e^{-x/\lambda}}{\Gamma(r)\,\lambda^{r}}\,,\,\,\, x>0\,,
\end{align*}
where $r>0$ is known and $\lambda >0$ unknown. 
Suppose that we are interested in estimating the future density of unobserved rv $Y$ with pdf $q_{\lambda}(y)=\frac{y^{r'-1} e^{-y/\lambda}}{\Gamma(r')\,\lambda^{r'}}\,,\,\,\, y>0$. Imagine that such an estimator can be obtained by $\hat{q}_1(\cdot \,; x)$. We use the Kullback--Leibler (KL) loss (equation \ref{KL}) in order to compare the efficiency of the proposed estimator with the actual $q(\cdot)$. 
\begin{align}
L_{KL}(q_{\lambda}, \hat{q}_1(\cdot))=\int_{\mathbb{R}} q_{\lambda}(y)\, \log \frac{q_{\lambda}(y)}{\hat{q}_1(y; x)}\, dy,\label{KL}
\end{align}
and the corresponding frequentist risk function is given as
\begin{equation}
\label{frequentistrisk}
R_{KL}(\lambda, \hat{q}_1) \,=\,  \int_{\mathbb{R}}  L_{KL}\left(q_{\lambda}(y), \hat{q}_1(\cdot)\right) \,p_{\lambda}(x) \, dx  \,.  
\end{equation} 

Previous studies (see, Corcuera and Giummol\`e 1999), indicates that under KL, the Bayes predictive density estimator for $Y$ based on observed $x$, prior and posterior $\pi(\cdot)$ and $\pi(\cdot \mid \x)$ respectively, is given as
\begin{equation}
\label{formula-mre}
\hat{q}_{\pi_{A}}(y; x) =  \int_{\Lambda}  q_{\lambda}(y)\,\pi(\lambda \mid x) \, d\lambda\,.
\end{equation}
The following contains some definitions and remarks that will be needed later on.
\begin{definition}\label{defs}
The pdf and the cumulative density function (cdf) of a rv $T$ follows the inverse--gamma distribution,  if we have 
    \begin{align}
        f_{a, b}(t)&=\frac{b^a}{\Gamma(a)}t^{-a-1}e^{-\frac{b}{t}}\,,\label{pdf1}\\
        F_{a,b}(t)&=\frac{\Gamma(a, \frac{b}{t})}{\Gamma(a)}\,, \label{cdf1}
\end{align}
where $\Gamma(m,n)$ is known as an upper incomplete gamma function and is defined as $\int_n^{\infty} t^{m_1} e^{-t}\,dt$.
\end{definition}
\begin{remark}
One can verify the following statements:
\begin{enumerate}
\item The marginal distribution of a random variable $\operatorname{IG}(s_1, s_2)$ associated with prior $\pi(\xi)$ is 
\begin{equation}\label{m2}
m(s_1, s_2)=\int_0^{\infty} \pi(\xi) \operatorname{IG}_{s_1, s_2}(\xi)\,d\xi\,.
\end{equation}
Note that if $\pi(\xi)=1/\xi$ for $\xi>0$, then 
\begin{equation}\label{m2flat}
m(s_1, s_2)=\frac{s_1}{s_2}.    
\end{equation}
\item \begin{equation}\label{m3}
m(s_1, s_2; \bar{\xi})=\int_0^{\bar{\xi}} \pi(\xi) \operatorname{IG}_{s_1, s_2}(\xi)\,d\xi\,,
\end{equation}
when $\pi(\xi)=1/\xi$ for $0\leq \xi \leq\bar{\xi}$, we have 
\begin{equation}\label{m3flat}
m(s_1, s_2; \bar{\xi})= \frac{\Gamma(s_1+1, \frac{s_2}{\bar{\xi}})}{s_2\,\Gamma(s_1)}\,. 
\end{equation}
Noting $\bar{\xi}\to \infty$,  $\Gamma(1+s_1, 0)=\Gamma(1+s_1)$ and hence (\ref{m3}) and (\ref{m3flat}) are equal.
\item \begin{equation}\label{c4}
C(k_1, k_2, s_1, s_2)=\int_0^{\infty} \pi(\xi_1)m(s_1, s_2; \xi_1)\operatorname{IG}_{k_1, k_2}(\xi_1)\,d\xi_1\,.
\end{equation}
In the case of $\pi(\xi)=\frac{1}{\xi_1\xi_2}$, for $\xi_1\geq \xi_2>0$, after some calculation we get
\begin{equation}
\frac{k_1 \, k_2^{k_1}\, s_2^{-k_1-2} \Gamma (k_1+s_1+2) \, _2\tilde{F}_1\left(k_1+1,k_1+s_1+2;k_1+2;-\frac{k_2}{s_2}\right)}{\Gamma (s_1)}\,,
\end{equation}
where $_2\tilde{F}_1$ is known as a regularized hypergeometric function.
\end{enumerate}
\end{remark}
\begin{definition}
A random variable $T$ is said to have a generalized beta prime density, whenever $T$ has the following distribution:
\begin{equation}
    \frac{1}{\operatorname{B(a,b)}}\frac{1}{\sigma}\frac{(t/\sigma)^{a \gamma -1}}{\left(1+(t/\sigma)^{\gamma}\right)^{a+b}}\,,
\end{equation}
where $\operatorname{B}(\cdot, \cdot)$ is the beta function, $\sigma>0$ is the scale parameter, and $a,b, \gamma>0$ are shape parameters. We denote it by $\operatorname{GB}'(a, b, \gamma, \sigma)$. $\operatorname{GB}'(a, b,\sigma, \gamma=1)$ is known as three parameter beta prime $\operatorname{B}'(a, b,\sigma)$. Beta prime (also known as a beta distribution of the second kind) is obtained by $\operatorname{B}'(a, b,\sigma=1)$ and denoted by $\operatorname{B}'(a, b)$
\end{definition}
\section{Bayes predictive density estimation}
The following provides the Bayes predictive density estimator under KL and prior density $\pi(\lambda)$.
\begin{lemma}\label{prime}
The posterior predictive density for $Y_1 \mid \lambda_1 \sim \operatorname{Gam}(r'_1, \lambda_1)$, based on $X_1 \mid \lambda_1  \sim \operatorname{Gam}(r_1, \lambda_1)$ and for for a prior $\pi(\lambda_1)$, is given by
\begin{align}\label{unres}
    \hat{q}_0(y_1; x_1)=\frac{1}{\operatorname{B}(r_1-1, r'_1)}\frac{m(r_1+r'_1-1, x_1+y_1)}{m(r_1-1, x_1)}\frac{x_1^{r_1-1}y_1^{r'_1-1}}{(x_1+y_1)^{r_1+r'_1-1}},\,\,\,\,\,y_1>0.
\end{align}
\end{lemma}
\begin{proof}
Using the notation of $\operatorname{Gam}_{r, \lambda}(t)$ for the pdf of rv $X$, which follows the gamma distribution with parameters of $r$ and $\lambda$, and the fact that $\operatorname{Gam}_{ r, \lambda}(x)=\frac{\lambda}{x} \operatorname{IG}_{r, x}(\lambda)=\frac{1}{r-1}\operatorname{IG}_{r-1, x}(\lambda)$, for $x>0, r>0, \lambda>0$, along with equation (\ref{m2}), one may verify that $p(x_1)=\int_0^{\infty} \frac{\pi(\lambda_1)}{\Gamma(r_1)\,\lambda_1^{r_1}}x_1^{r_1-1}e^{-\frac{x_1}{\lambda_1}}\,d\lambda_1=\frac{m(r_1-1, x_1)}{r_1-1}$. The posterior density is given by
\begin{equation*}
   \pi(\lambda_1 \mid x_1) =\frac{\pi(\lambda_1)\,\lambda_1^{-r_1}\, x_1^{-r_1-1}e^{-x_1/\lambda_1}}{\Gamma(r_1-1)\, m(r_1-1, x_1)}\,.
\end{equation*}
Thus the posterior predictive density is given by
\begin{align*}
    \hat{q}_0(y_1; x_1)&=\int_0^{\infty} q_{r'_1,\lambda_1}(y_1) \,\pi(\lambda_1 \mid x_1)\,d\lambda_1\\
    &=\frac{1}{\Gamma(r'_1)\,\Gamma(r_1-1)\,m(r_1-1, x_1)}\int_0^{\infty} \pi(\lambda_1)\,\lambda_1^{-(r_1+r'_1-1)-1}x_1^{r_1-1}y_1^{r'_1-1} e^{-(x_1+y_1)/\lambda_1}\,d \lambda_1\\
    &=\frac{1}{\operatorname{B}(r_1-1, r'_1)}\frac{m(r_1+r'_1-1, x_1+y_1)}{m(r_1-1, x_1)}\frac{x_1^{r_1-1}y_1^{r'_1-1}}{(x_1+y_1)^{r_1+r'_1-1}}\,,
\end{align*}
and hence the proof.
\end{proof}
\begin{remark}
Assuming the non--informative prior $\pi(\lambda)=\frac{1}{\lambda_1}$ for $\lambda_1 >0$ (and therefore $m(s_1, s_2)=s_1/s_2$) to Lemma \ref{prime}, the posterior predictive distribution can be simplified as
\begin{equation}\label{unr}
\frac{1}{\operatorname{B}(r_1, r'_1)}\frac{x_1^{r_1}y_1^{r'_1-1}}{(x_1+y_1)^{r_1+r'_1}}\,\,\,\,\,\,y_1>0,    
\end{equation}
which is known as a three parameter beta prime distribution, namely, $Beta' (r'_1, r_1, x_1)$, as first observed by Aitchison (1975).
L'moudden and Marchand (2017) also use Bayes predictive density estimator for a gamma model under Kullback–Leibler loss, when the scale parameter is bound in an interval $(a, b)$ for $b>a>0$, and show that restricted Bayes predictive density estimator dominates the unrestricted density estimator in (\ref{unr}).
\end{remark}
\subsection{Improving the Bayes predictive density estimation upon ancillary information}
Often, we have at our disposal additional information from which we may gain a greater understanding, if correctly employed.
In order to incorporate the ancillary information, let us assume the following scenario:

Let $X_i$ for $i=1,2$ be independent $\operatorname{Gam}(r_i, \lambda_i)$, 
$y_1$ has $\operatorname{Gam}(r'_1, \lambda_1)$, 
and $y_1$ and $X_i$, are independent where $r_i>0$, $r'_1$ are known and $\lambda_i>0$ is unknown but we have a prior information that the ratio of scale parameters satisfy $\frac{\lambda_1}{\lambda_2}\geq 1$ (i.e. $\lambda_1 \geq \lambda_2 >0$). The following theorem provides a posterior predictive density estimator subject to constraints placed on the ratio of scale parameters
by assuming that $\lambda_1$ and $\lambda_2$ are independent; that is, $\pi(\lambda_1, \lambda_2)=\pi(\lambda_1)\pi(\lambda_2)$. 
\begin{theorem}
The posterior predictive density of $Y_1 \sim \operatorname{Gam}(r'_1, \lambda_1)$, based on $X_1\sim \operatorname{Gam}(r_1, \lambda_1)$ and $X_2\sim \operatorname{Gam}(r_2, \lambda_2)$, where $r_i>0$, $r'_1>0$ and $\lambda_i>0$ are unknown associated with the prior information $\pi(\lambda_1, \lambda_2)=\pi(\lambda_1)\pi(\lambda_2)I_{(0, \infty)}(\lambda_1)\,I_{(0, \lambda_1)}(\lambda_2)$ satisfy $\frac{\lambda_1}{\lambda_2}\geq 1$, is given by
\begin{equation}\label{res}
    \hat{q}_1(y_1; x_1, x_2)=\frac{C(r_1+r'_1-1, x_1+y_1, r_2-1, x_2)}{C(r_1-1, x_1, r_2-1, x_2)}\frac{m(r_1-1, x_1)}{m(r_1+r'_1-1, x_1+y_1)}\hat{q}_0(y; x_1)\,,
\end{equation}
which is a weighted function of $\hat{q}_0(y; x_1)$, the unrestricted posterior predictive density of $y_1$  based on $x_1$ obtained in equation (\ref{unres}), and $m(\cdot, \cdot)$ and $C(\cdot, \cdot)$ are given in equations (\ref{m3}) and (\ref{c4}), respectively.
\end{theorem}
\begin{proof}
The joint posterior density is given by
\begin{equation}\label{joint}
\pi(\lambda_1, \lambda_2 \mid x_1, x_2)=\frac{\pi( \lambda_1)\pi(\lambda_2)\,p_{r_1, \lambda_1}(x_1) \,p_{r_2, \lambda_2}(x_2)}{\int_0^{\infty}\int_0^{\lambda_1}\pi(\lambda_1)\pi(\lambda_2)\,p_{r_1, \lambda_1}(x_1)p_{r_2,\lambda_2}(x_2)\,d\lambda_2 d\lambda_1}\,, 
\end{equation}
The denominator in equation (\ref{joint}) can be written as
$$\int_0^{\infty} \pi(\lambda_1)\,p_{r_1, \lambda_1}(x_1) \int_0^{\lambda_1} \pi(\lambda_2) p_{r_2, \lambda_2}(x_2)\, d\lambda_2 \,d\lambda_1\,.$$
Applying notations in (\ref{m3}) and (\ref{c4}), the inner integral in the above quotation is equal to 
$$\frac{1}{r_2-1} \int_0^{\lambda_1} \pi(\lambda_2) \operatorname{IG}_{r_2-1, x_2}(\lambda_2)\,d\lambda_2=\frac{m(r_2-1, x_2; \lambda_1)}{r_2-1}\,,$$
therefore the denominator is obtained by
\begin{align*}
    \frac{1}{(r_1-1)(r_2-1)} &\int_0^{\infty} \pi(\lambda_1)\,m(r_2-1, x_2; \lambda_1)\operatorname{IG}_{r_1-1, x_1}(\lambda_1)\,d\lambda_1\\
    &=\frac{C(r_1-1, x_1, r_2-1, x_2)}{(r_1-1)(r_2-1)}\,.
\end{align*}
Therefore the posterior density (\ref{joint}) has the following form:
\begin{equation}
    \frac{\pi(\lambda_1)\pi(\lambda_2)\frac{\lambda_1^{-r_1}}{\Gamma(r_1-1)}\frac{\lambda_2^{-r_2}}{\Gamma(r_2-1)}x_1^{r_1-1}x_2^{r_2-1}e^{-x_1/\lambda_1}e^{-x_2/\lambda_2}}{C(r_1-1, x_1, r_2-1, x_2)}\,,
\end{equation}
so the marginal posterior is given by
\begin{align*}
    \pi(\lambda_1 \mid x_1, x_2)&= \int_0^{\lambda_1} \pi(\lambda_1, \lambda_2 \mid x_1, x_2)\,d\lambda_2\\
    &=\frac{\pi(\lambda_1)\frac{\lambda_1^{-r_1}}{\Gamma(r_1-1)}x_1^{r_1-1}e^{-x_1/\lambda_1}}{C(r_1+-1, x_1, r_2-1, x_2)}\int_0^{\lambda_1}\pi(\lambda_2)\frac{\lambda_2^{-r_2}}{\Gamma(r_2-1)}x_2^{r_2-1}e^{-x_2/\lambda_2}\\
    &=\pi(\lambda_1)\frac{m(r_2-1, x_2; \lambda_1)}{C(r_1-1, x_1, r_2-1, x_2)} \, \frac{\lambda_1^{-r_1}}{\Gamma(r_1-1)}x_1^{r_1-1}e^{-x_1/\lambda_1}\,,\,\,\,\,\lambda_1>0\,.
\end{align*}
Substituting the above equation in the formula of posterior predictive density estimator, gives
\begin{align*}
\hat{q}_1(y_1; x_1, x_2)&=\int_0^{\infty} q_{r'_1, \lambda_1}(y_1)\,\pi(\lambda_1 \mid x_1, x_2)\,d\lambda_1\\
&=\frac{1}{\Gamma(r'_1)\,\Gamma(r_1-1)}\frac{x_1^{r_1-1} y_1^{r'_1}}{C(r_1-1, x_1, r_2-1, x_2)}\int_0^{\infty}\pi(\lambda_1)m(r_1-1, x_1)e^{-(x_1+y_1)/\lambda_1}\lambda_1^{-(r_1+r'_1-1)-1}\,d\lambda_1\\
&=\frac{1}{\operatorname{B}(r_1-1, r'_1)}\frac{C(r_1+r'_1-1, x_1+y_1, r_2-1, x_2)}{C(r_1-1, x_1, r_2-1, x_2)}\frac{x_1^{r_1-1} y_1^{r'_1}}{(x_1+y_1)^{r_1+r'_1-1}}\,,
\end{align*}
and replacing $\hat{q}_1(y_1; x_1)$ from (\ref{unres}) in the above equation thus completes the proof.
\end{proof}
\begin{corollary}
In the case of non informative prior $\pi(\lambda_1, \lambda_2)=\frac{1}{\lambda_1 \lambda_2}$, for $\lambda_1 \geq \lambda_2 \geq 0$, the Bayesian predictive density estimator given in equation (\ref{res}) is in fact the weighted beta prime distribution and is given by
\begin{align}\label{q1}
\hat{q}_1(y_1; x_1, x_2)=\frac{r_1 x_2^{-r} y_1^r ,\Gamma (r'+r_1+r_2) \, _2F_1\left(r'+r_1,r'+r_1+r_2;r'+r_1+1;-\frac{x_1+y_1}{x_2}\right)}{(r'+r_1) \Gamma (r') \Gamma (r_1+r_2) \, _2F_1\left(r_1,r_1+r_2;r_1+1;-\frac{x_1}{x_2}\right)}  \,.
\end{align}
\end{corollary}

\section{Real examples and illustrations}\label{real}
In this section, we apply the proposed methods in Section 3 in order to find the Bayes predictive density estimators based on the National Hockey League’s dataset for the $2017-2018$ season to predict the upcoming $2018-2019$ season.
We are concerned with the waiting time prediction until scoring $r$-th goal (in a real scenario $r$ rarely exceeds $6$) by team A as compared to team B. Similar explanations have been done by Suzuki et al. (2010) in predicting match outcomes football World Cup 2006. We shall assume that $N_1$ and $N_2$, the number of goals scored by team A and B, are two independent, random Poisson variables such that
\begin{align*}
    N_1 \mid \theta_1 \sim \operatorname{Po}\left(\theta_1 \frac{R_1}{R_2}\right)\,,\,\,\,
    N_2 \mid \theta_2 \sim \operatorname{Po}\left(\theta_2 \frac{R_2}{R_1}\right)\,,
\end{align*}
where $\theta_1$ and $\theta_2$ can be interpreted as the mean number of goals teams A and B score against their opponents, and $R_1$ and $R_2$ are the teams A and B points in the previous $2017-2018$ season. This assumption seems quite reasonable, since it is expected a that team with a higher point scores more goals rather than lower point team.

Equivalently, we can assume the waiting time to see team A scores $r_1$ goals to team B and team B scores $r_2$ goals to team A independently follow
\begin{align}
    X_1 &\mid \theta_1, r_1 \sim \operatorname{Gam}\left(r_1, \theta_1 \frac{R_1}{R_2}\right)\,,\\
    X_2 &\mid \theta_2, r_2 \sim \operatorname{Gam}\left(r_2, \theta_2 \frac{R_2}{R_1}\right)\,.
\end{align}
The goal is to estimate the future density for the expected time to seeing $r_1$ goal by team A to team B, $Y_1 \mid \theta_1, r'_1 \sim \operatorname{Gam}\left(r'_1, \theta_1\frac{R'_1}{R'_2}\right)$, using ancillary information $\theta_1\geq \theta_2$ (i.e., by contemplating that team A is more capable of scoring against team B based on previous records and$/$or the predictions of sports analysts.

As an example, we consider two Canadian hockey team, \textit{the Toronto Maple Leafs (A)} and \textit{Montreal Canadiens (B)}. Based on the 2017--18 season teams' points, available at \href{https://www.nhl.com}{\tt www.nhl.com}, $R_1=105$ and $R_2=71$, respectively.
Table 1 (on page 10) shows the time elapsed (in minutes) until scoring the third goal in games ($r_1=r_2=3$) played by the Toronto Maple Leafs versus the Montreal Canadiens in the 2017--2018 National Hockey League season.

In Figure \ref{q0q1}, the Bayesian predictive density estimators $\hat{q}_0$ and $\hat{q}_1$ (based on equations \ref{unr} and \ref{q1} with the exact densities given in Table 2) for the elapsed time until the Toronto Maple Leafs scores the $3^{rd}$ goal to the Montreal Canadiens in an upcoming game are depicted. Note that $\hat{q}_1$ is obtained by employing ancillary information that the performance of the Toronto Maple was better the Montreal Canadiens (specialists' opinions or their points) i.e. $\theta_1/\theta_2 \geq 1$ as well as contemplating their points.
It is important to mention here that the densities in here have been truncated to $(0, 60)$ minutes since NHL hockey games consist of three periods, each lasting $20$ minutes.
\begin{figure}[H]
    \centering
    \includegraphics[width=0.5\textwidth]{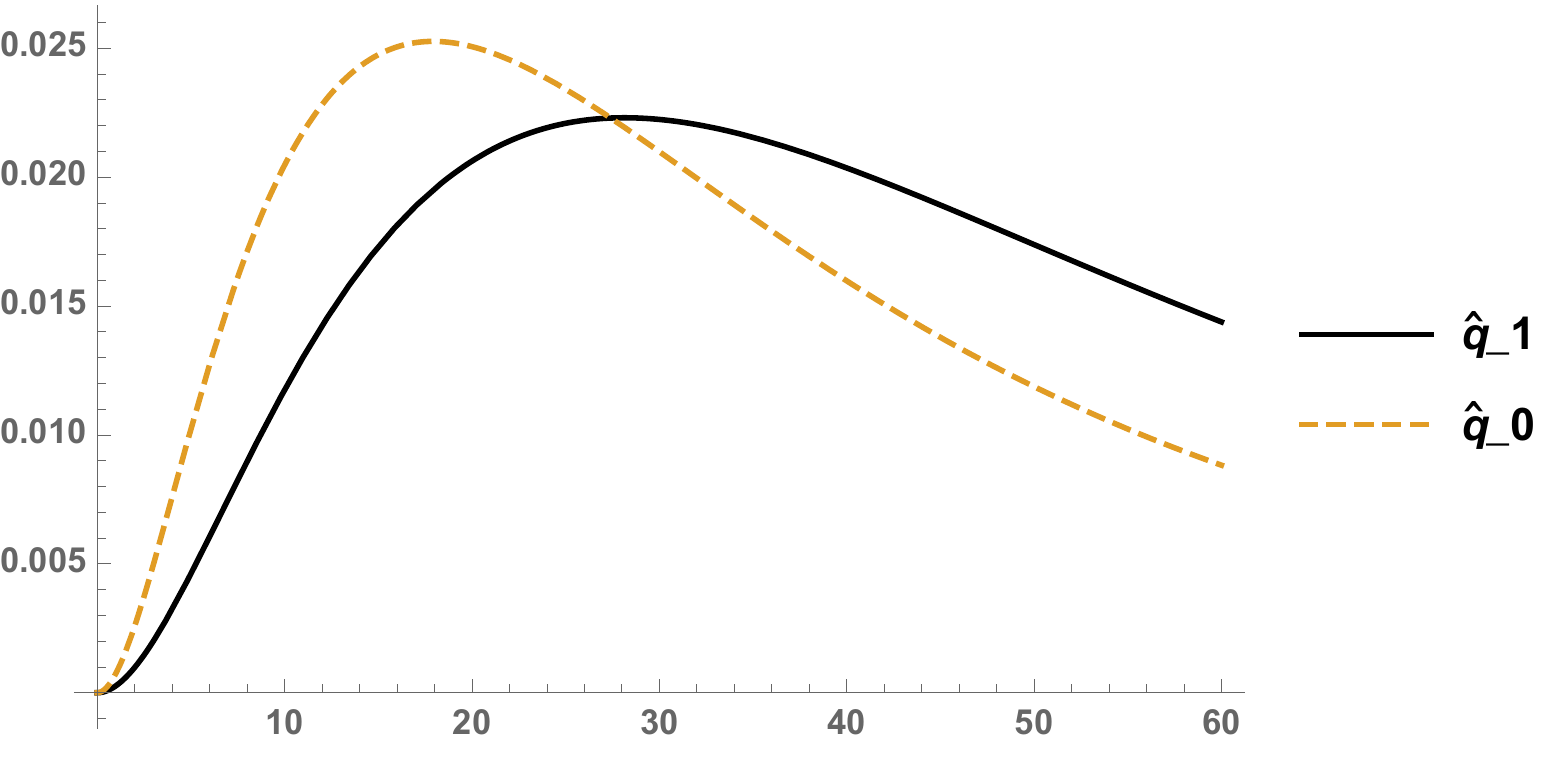}
    \caption{Density of waiting time until scoring the $3^{rd}$ goal by the Toronto Maple Leafs versus the Montreal Canadiens. $\hat{q}_0$ is the Bayes pdf without ancillary information and $\hat{q}_1$ includes ancillary information, such as specialists' opinions and last year's points, based on $r_1=r_2=r'=3$ and the data presented in Table 1. }
    \label{q0q1}
\end{figure}
Table 2 contains the density of obtained estimator $\hat{q}_{0}$ and the Bayes estimators $q_{1}$, along with their modes, means, and $10^{th}$, $50^{th}$ and $90^{th}$ percentiles for the future density $Y_1$ of waiting time until scoring $3^{rd}$ score by the Toronto Maple Leafs based on the data from Table 1.

For instance, if we contemplate the ancillary information as we described above, we are anticipating to see the third goal of the Toronto Maple Leafs around minute $28$, rather than minute $18$ of the game, or we are expecting on average the Toronto Maple Leafs score the third goal versus their opponents around minute $28$ of each game. However, we apply the information at our disposal to improve our estimation, we hope to see the third goal of the Toronto Maple Leafs versus the Caniadiens around minute $33$ of the game, or in less than $20$ percent of matches Toronto scores the third goal at $14.38$. By applying the available information, this time moves to end of first period of the game. 
\begin{table}[H]
\begin{tabular}{|c|c|c|}
\hline
\rowcolor[HTML]{EFEFEF} 
Estimator & Truncated Density on $(0, 60)$& Mode,\, Mean,\, $P_{0.20}$,\,\, $P_{0.50}$,\,\, $P_{0.90}$  \\ \hline
$\hat{q}_0$ & $1901470 y_1^2/(35.85 + y_1)^6 $& 17.92\, 28.35, 14.38, 26.62, 50.3 \\ \hline
$\hat{q}_1$ & $y_1^2 (0.055 + (0.0004 + 10^{-6} y_1) y_1)/(1.92 + 
  0.025 y_1)^8 $& 28.13,\, 33.12 19.06, 32.82, 53.48 \\  \hline

\end{tabular}
\caption*{Table 2: Bayesian predictive density estimators $\hat{q}_1$, $\hat{q}_0$, with and without ancillary information respectively, along with the truncated probability density function on $(0, 60)$ minutes as well as the mode, mean, and $P_i$, $i \times 100$ percentiles.}
\label{my-label}
\end{table}
\subsection{Prediction errors}
We calculate the the prediction error (pe) of $\hat{q}_0$ and $\hat{q}_1$ in order to investigate their performance under KL loss given in equation (\ref{KL}). We verify their pe's based on the $2016-17$ season and see their performance by comparing the KL loss (distance) Bayesian density estimator and exact density for waiting time until the Toronto Maple Leafs scores the $3^{rd}$ goal. This can be calculated by
\begin{align*}
pe(\hat{q}_i)=\mathbb{E}^{Y_1} \log \frac{q_{\lambda}(y_1)}{\hat{q}_i(y_1)}\,,
\end{align*} 
where $Y_1$, waiting time until the Toronto Maple Leafs scores the $3^{rd}$ goals based on the $2016-17$ dataset, follows a truncated gamma $\operatorname{Gam}(3, 18.3)$ on $(0, 60)$ (with an expected value of $35.8$).
It can be verified that the prediction error of our estimator when the ancillary information is taken into account $pe(\hat{q}_1)=0.04$ while $pe(\hat{q}_0)=0.45$ and $\hat{q}_1(y_1)$ outperforms other estimator, in estimating the density of waiting time until the $3^{rd}$ goal occurs in a match Toronto Maple Leafs versus the Montreal Canadiens in a National Hockey League's game.

However, the frequentist risks of  $\hat{q}_1(y_1)$ and  $\hat{q}_0(y_1)$, from equation (\ref{frequentistrisk}), can be used to verify the dominance result of our proposed estimator $\hat{q}_1(y_1)$ in order to make a prediction the result of upcoming match the Toronto Maple Leafs versus the Montreal Canadiens in the $2018-19$ season based on data from the $2017-18$ season.  
Figure \ref{riskq0q1} depicts the performance of $\hat{q}_1(y_1)$ and  $\hat{q}_0(y_1)$ via the frequentist risk functions. Since there is no ancillary information (i.e. $\theta_1=\theta_2$ and $R_1/R_2=1$), both risks are equal and as long as the ratio of $\theta_1/\theta_2$ increases (gain from ancillary information), the risk of $\hat{q}_1(y_1)$ decreases (the risk of $\hat{q}_0[y_1]$ is constant as we anticipated, since it is an MRE estimator) and both risk functions converge to each other due to lack of information.
\begin{figure}[htp]
    \centering
    \includegraphics[width=0.7\textwidth]{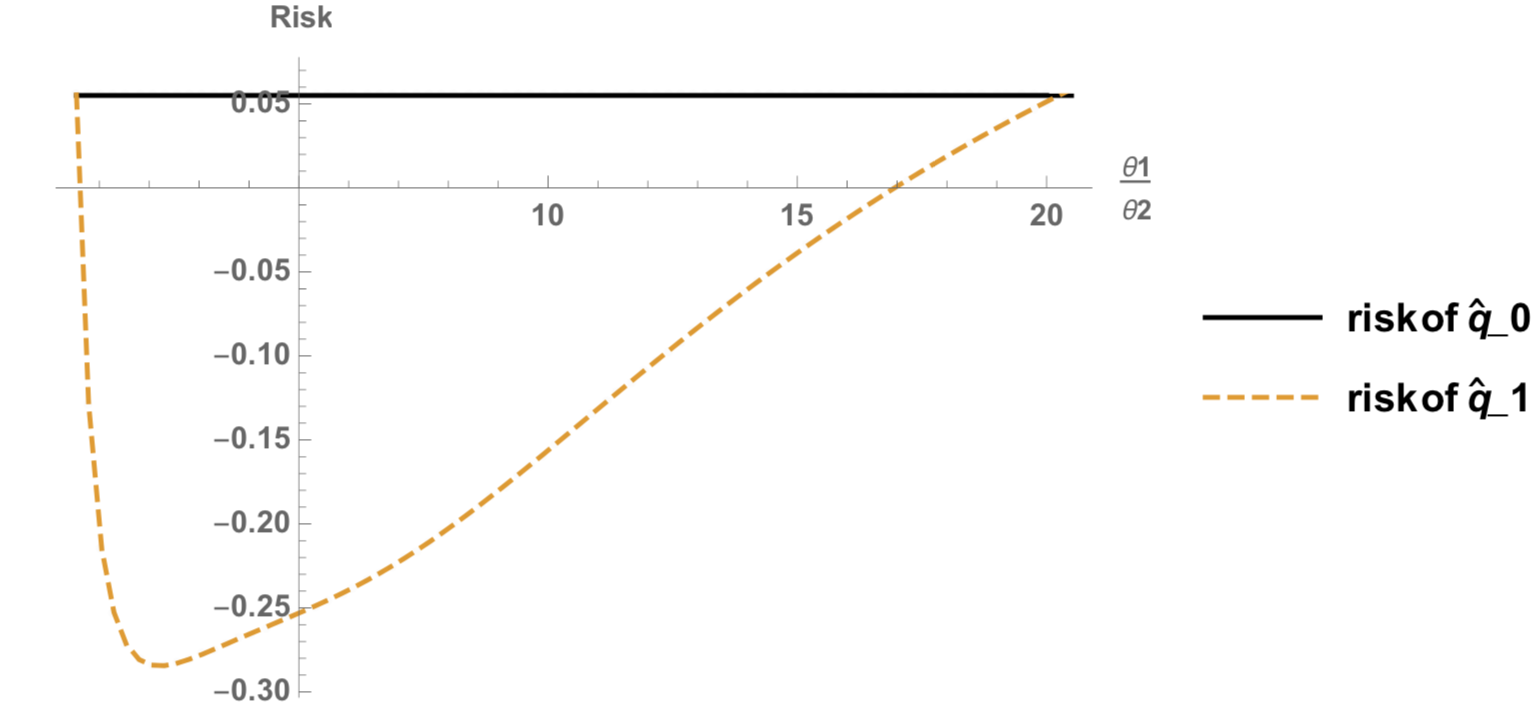}
    \caption{Frequentist risk functions of $\hat{q}_0$ and $\hat{q}_1$ based on $r_1=r_2=r'=3$ and the data presented in Table 1. }
    \label{riskq0q1}
\end{figure}
\begin{table}[h]
    \begin{minipage}{.55\linewidth}
    \centering
        \begin{tabular}{lc}
\cellcolor[HTML]{EFEFEF} Toronto Maple Leafs v.s. Opponents &\cellcolor[HTML]{EFEFEF}  Time Elapsed \\
\hline
Winnipeg Jets & 18.38 \\
New York Rangers & 11.12 \\
Chicago Blackhawks & 55.68 \\
New Jersey Devils & 53.57 \\
Montreal Canadiens & 32.67 \\
Detroit Red Wings & 15.75 \\
Ottawa Senators & 52.85 \\
Los Angeles Kings & 42.88 \\
Carolina Hurricanes & 27.17 \\
Anaheim Ducks & 58.48 \\
St. Louis Blues & 50.13 \\
Vegas Golden Knights & 15.03 \\
Minnesota Wild & 43.65 \\
Boston Bruins & 46.83 \\
Montreal Canadiens & 40.4 \\
Carolina Hurricanes & 31.6 \\
Calgary Flames & 41.88 \\
Edmonton Oilers & 13.08 \\
Pittsburgh Penguins & 12.9 \\
Carolina Hurricanes & 10.53 \\
Buffalo Sabres & 27.63 \\
New York Rangers & 31.35 \\
Arizona Coyotes & 11.4 \\
Montreal Canadiens & 57.17 \\
Buffalo Sabres & 57.93 \\
Pittsburgh Penguins & 31.57 \\
Colorado Avalanche & 58.07 \\
Vegas Golden Knights & 40.43 \\
Dallas Stars & 45.18 \\
Buffalo Sabres & 26.35 \\
Montreal Canadiens & 39.5 \\
Ottawa Senators & 52.43 \\
Tampa Bay Lightning & 35.67 \\
Ottawa Senators & 44.33 \\
Dallas Stars & 29.45 \\
New York Islanders & 30.52 \\
New York Rangers & 20.83 \\
Anaheim Ducks & 35.43 \\
Ottawa Senators & 11.48 \\
Tampa Bay Lightning & 31.57 \\
Columbus Blue Jackets & 28.05 \\
Pittsburgh Penguins & 35.73 \\
Detroit Red Wings & 59.48 \\
New York Islanders & 56.48 \\
Boston Bruins & 39.05 \\
Tampa Bay Lightning & 45.42 \\
Detroit Red Wings & 47.43 \\
Florida Panthers & 13.9 \\
New York Islanders & 37.28 \\
Nashville Predators & 36.72
\end{tabular}
    \end{minipage} 
        \begin{minipage}{.35\linewidth}
        \centering
        \begin{tabular}{lc}
\cellcolor[HTML]{EFEFEF} Canadiens v.s. Opponents & \cellcolor[HTML]{EFEFEF}Time elapsed \\
\hline
Toronto Maple Leafs & 31.52 \\
Florida Panthers & 38.3 \\
New York Rangers & 13.23 \\
Ottawa Senators & 13.27 \\
Minnesota Wild & 54.78 \\
Winnipeg Jets & 48.27 \\
Vegas Golden Knights & 23.52 \\
Arizona Coyotes & 35.3 \\
Buffalo Sabres & 48.43 \\
Columbus Blue Jackets & 58.57 \\
Detroit Red Wings & 25.47 \\
Detroit Red Wings & 21.83 \\
St. Louis Blues & 46.55 \\
Vancouver Canucks & 37.05 \\
Winnipeg Jets & 50.23 \\
Calgary Flames & 43.15 \\
Detroit Red Wings & 34.62 \\
New York Islanders & 30.48 \\
New Jersey Devils & 54.67 \\
Dallas Stars & 29.25 \\
Pittsburgh Penguins & 32.93 \\
Vancouver Canucks & 48.73 \\
Boston Bruins & 28.82 \\
Pittsburgh Penguins & 34.38 \\
New York Islanders & 39.25 \\
Washington Capitals & 58.68 \\
Colorado Avalanche & 52.57 \\
Carolina Hurricanes & 28.98 \\
Anaheim Ducks & 10.2 \\
Ottawa Senators & 38.75 \\
Philadelphia Flyers & 57.1 \\
Vegas Golden Knights & 48.72 \\
New York Rangers & 58.68 \\
Tampa Bay Lightning & 36.47 \\
New York Islanders & 33.58 \\
Buffalo Sabres & 59.25 \\
Washington Capitals & 49.47 \\
Detroit Red Wings & 29.47\\
\hline
\end{tabular}
\caption*{Table 1: Time elapsed (in minutes) until scoring the third goal in games corresponding to Toronto Maple Leafs (left) and Montreal Canadiens (right). }
    \end{minipage}%
\end{table}

\section{Concluding remarks}\label{conclude}
In summation, our results demonstrate that applying additional ancillary information can be used to come up with a more accurate prediction via a better Bayesian predictive density estimator for the future density of a gamma. Our proposed density estimator for the waiting time until scoring the $r$-th goal outperforms the usual MRE estimator. This was verified by comparing the prediction errors of two density estimators as well as their frequentist risk functions.

\section*{Acknowledgement}
The authors thank the Stathletes company specially Meghan Chayka and Jeff Goeree for providing the data and helpful comments to this manuscript. The Natural Sciences and the Engineering Research Council of Canada, and Ontario Centre of Excellence supported the research of Professor S. Ejaz Ahmed.

\bibliographystyle{plain}
\bibliography{dissertation}

\end{document}